\documentclass[runningheads,a4paper]{llncs}

\usepackage[utf8]{inputenc}
\usepackage[tbtags]{amsmath}
\usepackage[subnum]{cases}
\usepackage{amssymb}

\usepackage{amsthm,hyperref}
\usepackage{nomencl}
\usepackage{enumitem}
\usepackage{graphicx,wasysym,color,subfigure,lscape, threeparttable, textgreek}
\usepackage[a4paper]{geometry}
\usepackage{etoolbox}
\usepackage{multirow}
\usepackage{booktabs}
\usepackage{rotating}
\usepackage{algorithm,algorithmic}
\usepackage[algo2e,vlined,ruled]{algorithm2e}
\usepackage{tikz}
\newcommand{\Mypm}{\mathbin{\tikz [x=1.4ex,y=1.4ex,line width=.1ex] \draw (0.0,0) -- (1.0,0) (0.5,0.08) -- (0.5,0.92) (0.0,0.5) -- (1.0,0.5);}}%

\newtheorem{thm}{Theorem}[section]
\newtheorem{lem}[thm]{Lemma}
\usepackage[utf8]{inputenc}

\usepackage{url}

\urldef{\mailsa}\path|{harsha, emreyamangil, rbent}@lanl.gov|
\urldef{\mailsb}\path|,mlu87@g.clemson.edu|
\newcommand{\keywords}[1]{\par\addvspace\baselineskip
\noindent\keywordname\enspace\ignorespaces#1}

\begin{document}
\mainmatter
\title{Tightening McCormick Relaxations for Nonlinear Programs via Dynamic Multivariate Partitioning}

\titlerunning{Dynamic McCormick partitioning with bound tightening}

\author{Harsha Nagarajan$^{\dag}$ \and Mowen Lu$^\ddag$ \and  Emre Yamangil$^\dag$ \and Russell Bent$^{\dag}$}
\authorrunning{Dynamic McCormick partitioning with bound tightening}


\institute{$^{\dag}$ Center for Nonlinear Studies, Los Alamos National Laboratory, NM, USA\\
$^{\ddag}$ Department of Industrial Engineering, Clemson University, SC, USA\\
\mailsa
\mailsb}
% \url{http://www.springer.com/lncs}}


% \toctitle{Lecture Notes in Computer Science}
% \tocauthor{Authors' Instructions}

\maketitle

\begin{abstract}
In this work, we propose a two-stage approach to strengthen piecewise McCormick relaxations for mixed-integer nonlinear programs (MINLP) with multi-linear terms. In the first stage, we exploit Constraint Programing (CP) techniques to contract the variable bounds. In the second stage we partition the variables domains using a dynamic multivariate partitioning scheme. Instead of equally partitioning the domains of variables appearing in multi-linear terms, we construct sparser partitions yet tighter relaxations by iteratively partitioning the variable domains in regions of interest. This approach decouples the number of partitions from the size of the variable domains, leads to a significant reduction in computation time, and limits the number of binary variables that are introduced by the partitioning.  We demonstrate the performance of our algorithm on well-known benchmark problems from MINLPLIB and discuss the computational benefits of CP-based bound tightening procedures. 
\end{abstract}


%==============================;
%  Include all Sections here   ;
%==============================;
% Introduction 
\section{Introduction}
Mixed Integer Nonlinear Programs (MINLPs) are part of a class of non-convex, mathematical programs that include discrete variables and nonlinear terms in the objective function and/or constraints. Within many application domains, MINLPs with multi-linear, non-convex terms are of great interest. For example, these problems appear in chemical engineering (synthesis of process/water networks) \cite{meyer2006global,ryoo1995global}, energy infrastructure networks \cite{coffrin2015strengthening}, and in the molecular distance geometry problem \cite{liberti2008branch}. 
Despite their importance in such areas, these problems remain difficult to solve.
Global optimization solvers, like BARON \cite{sahinidis1996baron}, depend heavily on the quality of mixed-integer linear programing relaxations to MINLPs. 
However, these relaxations are often weak and the solvers are not guaranteed to converge to a global optimum or even find a feasible solution. 
As a result, there is considerable interest in developing tighter relaxations that improve the convergence of global solvers. In this paper we focus on MINLPs with multi-linear terms, though the approach is generalizable.

In the context of this paper, there are two key methods for deriving tight relaxations of MINLPs with multi-linear terms. First, variable bounds are a critical contributor to the quality of relaxations. As a result, bound tightening methods have received a great deal of attention, in particular for problems with bilinear terms \cite{castro2015tightening,belotti2012feasibility,castro2015normalized,faria2011novel,mouret2009tightening}. 
In most of these papers, the most common approaches solve sequences of minimization and maximization problems where the continuous variables are the objective.  The solutions to these problems are used to tighten the bounds of the variables. In this paper, we combine these bound tightening approaches with constraint programming to improve their effectiveness.
The second method for tightening relaxations focuses on reducing the size of the relaxed feasible space. This is often done with domain partitioning, i.e. spatial branch-and-bound (sBB). In sBB, a single variable (often the variable that violates the feasible region the most) is iteratively partioned within a branch-and-bound search tree \cite{tawarmalani2005polyhedral,smith1999symbolic}.
One of the crucial requirements of successful sBB is tight lower bounds on the objective. These bounds support efficient pruning of infeasible regions and some of the most effective bounds are those that are based on relaxations.
When multi-linear terms are involved, a commonly used method is McCormick relaxations. As McCormick relaxations tend to be loose in many situations, the literature contains many efforts to improve these relaxations. The method most closely related to our own builds uniform piecewise McCormick relaxations via univariate or bivariate partitioning \cite{castro2015tightening,hasan2010piecewise,karuppiah2006global}.
One of the drawbacks of such approaches is that they may need a large number of partitions that are controlled by on/off binary variables. As these binary variables introduce combinatorial inefficiencies, this approach is often restricted to small problems.
To address this issue, there has been recent work focusing on addressing these inefficiencies.
For example, \cite{castro2015normalized,castro2014optimality} combines multiparametric disaggregation with optimality-based bound tightening methods. In \cite{wicaksono2008piecewise}, the authors discuss a non-uniform, bivariate partitioning approach that improves relaxations but provide results for a single, simple benchmark problem. More recently, in \cite{hasan2010piecewise}, the authors report the advantages of bivariate (as compared to univariate) partitioning, however they use partitions chosen at uniformly located grid points. In the context of multi-linear terms, \cite{teles2013univariate} discusses a univariate parametrization method that solves medium-sized benchmarks. However, none of these approaches address the key limitation of uniform partitioning, \emph{partition density}, i.e. these methods introduce partitions in unproductive regions of the search space.   
We address this limitation by introducing an approach that dynamically partitions the relaxations in regions of the search space that favor optimality.
To the best of our knowledge, there is little or no work on methods for solving MINLPs with multi-linear terms with such sparse partitioning. 

To summarize, we address the problem of tight relaxations for non-convex multi-linear functions and we develop a two-stage algorithm that strengthens piecewise multi-linear relaxations. In the first stage, we apply a sequential, CP inspired, bound tightening approach. In the second stage, we develop a dynamic, sparse multivariate partitioning approach that addresses the key limitations of uniform partitioning approaches. 
With this algorithm, we are able to solve many MINLPs more efficiently and accurately with fewer parameter tuning options than the existing approaches. 
The remainder of this paper is organized as follows: Section \ref{sec:problem} discusses the required notation, problem set up and reviews McCormick relaxations for bilinear and multi-linear terms. Section \ref{sec:algo} discusses a sequential bound tightening approach, formalizes the concepts and notation for piecewise relaxations of McCormick envelopes, and provides a detailed discussion on multivariate dynamic partitioning algorithm on multi-linear and monomial terms. Section \ref{sec:results} illustrates the strength of the proposed algorithms on benchmark MINLPs and Section \ref{sec:conclusions} concludes the paper.

%!TEX root = ../tmc.tex
% Notation and Problem definition

\section{Problem definition}
\label{sec:problem}
\newcommand{\calP}{\boldsymbol{\mathcal{P}}}

\noindent\textbf{Notation:} Here, we use lower and upper case for vector and matrix entries, respectively. Bold font refers to the entire vector or matrix.
With this notation, $||\mathbf{v}||_2$ defines the $L_2$ norm of vector $\mathbf{v}\in\mathbb{R}^n$. Given vectors $\mathbf{v}_1 \in\mathbb{R}^n$ and $\mathbf{v}_2 \in\mathbb{R}^n$, $\mathbf{v}_1 \cdot \mathbf{v}_2 = \sum_{i=1}^n {v_1}_i {v_2}_i$; $\mathbf{v}_1 + \mathbf{v}_2$ implies element-wise sums; and $\frac{\mathbf{v}_1}{\alpha}$ denotes the element-wise ratio between entries of $\mathbf{v}_1$ and the scalar $\alpha$. Next, $z \in \mathbb{Z}^+$ represents a strictly positive integer scalar.  
$\mathbf{M} \in \mathbb{S}^{n\times n}$ represents a symmetric square matrix $\mathbf{M}$. Given variables $x_i$ and $x_j$, $\langle x_i,x_j\rangle^{MC}$, $\langle x_i,x_j\rangle^{UTMC}$ and $\langle x_i,x_j\rangle^{DTMC}$ denote the McCormick envelope, uniformly-partitioned McCormick envelope and dynamically-partitioned McCormick envelope, respectively. $(x_i^L,x_i^U)$ denotes the prescribed global lower and upper bound and $(x_i^l,x^u_i)$ denotes the tightened lower and upper bound, respectively. \\
 
\noindent\textbf{Problem} The problems considered in this paper are MINLPs, where the non-linearity is due to multi-linear (polynomial) functions. Often, these problems are not convex. The general form of the problem, denoted as $\calP_0$, is as follows: 

{\fontsize{9}{8}\selectfont
\begin{equation*}
\begin{aligned}
\calP_0: \ \ \ & \underset{\mathbf{x},\mathbf{y}}{\text{minimize}} & &  f(\mathbf{x},\mathbf{y},\mathbf{z}) \\
&\text{subject to} & & g(\mathbf{x},\mathbf{y},\mathbf{z}) \leq 0, \\
& & & h(\mathbf{x},\mathbf{y},\mathbf{z}) = 0, \\
& & & z_{K} = x_ix_j\ldots x_k, \ \ \forall K \in ML\\
& & & \mathbf{x}^L \leq \mathbf{x} \leq \mathbf{x}^U, \\
& & & \mathbf{y} \in \{0,1\}^m
\end{aligned}
\end{equation*}}

\noindent where, $f: \mathbb{R}^n \rightarrow \mathbb{R}$ is a scalar multi-linear function and $g: \mathbb{R}^n \rightarrow \mathbb{R}^{m_1}, h: \mathbb{R}^n \rightarrow \mathbb{R}^{m_2}$ are vector, multi-linear functions. $\mathbf{x},\mathbf{y}$ and $\mathbf{z}$ are vectors of continuous variables with box constraints, binary variables, and multi-linear functions, respectively. $z_K$ is the $K^{th}$ multilinear term in the set $ML$ such that $ML = \{K=(i,j,\ldots,k) | z_K = x_ix_j\ldots x_k\}$. When $i=j=\ldots=k$, the multi-linearity is reduced to monomial terms. 

\subsection{Standard convex relaxations for multi-linear terms}
\label{subsec:relaxations}
\noindent\textbf{McCormick relaxations} Given two variables, $x_i$,$x_j\in \mathbb{R}$ such that $x_i^l \leq x_i \leq x_i^u$ and $x_j^l \leq x_j \leq x_j^u$, we define the McCormick relaxation \cite{mccormick1976computability} of the bilinear product $x_ix_j$ as 
$\widehat{x_{ij}} \in \langle x_i,x_j \rangle^{MC}$ such that $\widehat{x_{ij}}$ satisfies
\vspace{-0.4cm}
{\fontsize{9}{8}\selectfont
\begin{subequations}
\begin{align}
   \widehat{x_{ij}} &\geq x_i^l x_j + x_j^l x_i - x_i^lx_j^l\\ 
   \widehat{x_{ij}} &\geq x_i^u x_j + x_j^u x_i - x_i^ux_j^u\\ 
   \widehat{x_{ij}} &\leq x_i^l x_j + x_j^u x_i - x_i^lx_j^u\\ 
   \widehat{x_{ij}} &\leq x_i^u x_j + x_j^l x_i - x_i^ux_j^l 
\end{align}
\label{eq:mcc}
\end{subequations}}
%\vspace{-0.1cm}

\noindent
The relaxations in \eqref{eq:mcc} are exact when one of the variables involved in the product is a binary variable. Further, relaxations in \eqref{eq:mcc} can be reduced to a simpler form (three constraints) when both the variables involved in the product are binary variables. If $y_i$ and $y_j$ are binary variables, we denote this simplified relaxation as $ \widehat{y_{ij}} \in \langle y_i,y_j \rangle^{BMC}$.
\\ 

\noindent\textbf{Successive McCormick relaxations of multi-linear terms}
Given a multi-linear term \\$x_ix_j
\ldots x_k$ with $k$-linear terms, we use a general technique for successively deriving McCormick envelopes on bilinear combinations of the terms. %as described earlier in this section. 
As discussed in \cite{cafieri2010convex}, the tightness of McCormick relaxations depends on the grouping order of bilinear terms. 
Here, we assume a lexicographic order of grouping the bilinear terms. For example, given a multi-linear term $(x_1x_2x_3x_4)$, the successive ordering of bilinear terms is $(((x_1x_2)x_3)x_4)$. More formally, for $k$-linear terms, the McCormick envelope of $x_ix_j\ldots x_{k-1}x_k$ is represented as 
$$\langle x_ix_j\ldots x_{k-1} x_k\rangle^{MC} = \langle\langle\langle x_ix_j\rangle^{MC}\ldots x_{k-1}\rangle^{MC} x_k \rangle^{MC}.$$ Study of alternate grouping choices is beyond the scope of this paper and is a topic of future work.

%!TEX root = ../tmc.tex
% Proposed algorithm

\section{CP-DTMC Algorithm}
\label{sec:algo}

The Constraint Programming with Dynamic Tightening of McCormicks (CP-DTMC) algorithm is described in this section. It combines CP based domain tightening with a partitioning scheme for McCormick relaxations.

\subsection{Sequential bound tightening procedure}
\label{subsec:CP}
%{\color{blue}Russell, here is where we may need more text in context of CP jargons?}
\vspace{-0.3cm}
The first stage of CP-DTMC tightens the bounds of the continuous variables of $\calP_0$.
In many engineering applications there is little or no information about the upper and lower bounds ($\mathbf{x}^L,\mathbf{x}^U$) of these variables. Even when known, the gap between the bounds is often large. 
As discussed earlier, these bounds are used in McCormick relaxations to derive convex envelopes of  multi-linear terms in $\calP_0$. 
Large bounds generally weaken these relaxations, degrade the quality of the lower bounds, and slow the convergence of branch-and-cut algorithms. 
In practice, replacing the original bounds with tighter bounds can (sometimes) dramatically improve the quality of these relaxations (see Figure \ref{fig:mcc_region}[a]).

The basic idea of bound tightening is to derive (new) valid bounds to improve the relaxations. Our approach is
based on the work \cite{castro2015tightening} and is related to the iterative bound tightening of \cite{coffrin2015strengthening}. Let $x_{i}, i=1,\ldots,n$ be the element-wise entries of a continuous variable vector $\mathbf{x} \in \mathbb{R}^n$. In order to shrink the bounds of $x_{i}$, we solve a modified version of $\calP_0$. For each $x_{i}$, we first solve $\calP_0$ where we minimize 
$x_{i}$ and then solve $\calP_0$ where we maximize $x_{i}$. In both cases we add a constraint that bounds the original objective function of $\calP_0$ with a best known feasible solution
$(\mathbf{x}^*_{loc},\mathbf{y}^*_{loc},\mathbf{z}^*_{loc})$. This is a key difference between our approach and \cite{castro2015tightening}, \cite{coffrin2015strengthening}.
We also iteratively tighten the domain (bounds) of the variables using the approach above.
While there are other CP propagation methods that could be used to further improve the quality of the bounds, this method was sufficient to demonstrate the effectiveness of the overall approach.

More formally, Algorithm \ref{Algo:bound} describes the first stage of CP-DTMC. Line 1 takes as input the current bounds and a feasible solution. 
The core of the algorithm is embedded in Line 4. This is where we solve the variations of  $\calP_0$. Line 4a states the minimization and maximization of
$x_i$. Line 4b adds a bound on the original objective function. Lines 4c-4f state the rest of $\calP_0$.
Based on these solutions, we update the bounds of our variables (line 5).  The procedure continues until the bounds do not change (line 2).
Algorithm \ref{Algo:bound} is naturally parallel as each MILP of line 4 is independently solvable.

%----------------------------------------------;
%  Algorithm: Bound contraction algorithm      ;
%----------------------------------------------;
\begin{algorithm}[h]
\footnotesize
%\caption{\textbf{}}
\caption{Sequential bound tightening on $\mathbf{x}$ vector}
\label{Algo:bound}
\begin{algorithmic}[1]
\STATE Input: $\mathbf{x}^l \gets \mathbf{x}^L, \mathbf{x}^u \gets \mathbf{x}^U$,$\mathbf{x}^l_{iter} = \mathbf{x}^u_{iter} \gets \mathbf{0}$, $\mathbf{x}^*_{loc},\mathbf{y}^*_{loc},\mathbf{z}^*_{loc}$, $TOL>0$.  
\WHILE{$||\mathbf{x}^l - \mathbf{x}^l_{iter}||_2 > TOL$ and $||\mathbf{\mathbf{x}}^u - \mathbf{x}^u_{iter}||_2 > TOL$}
\STATE $\mathbf{x}^l_{iter} \gets \mathbf{x}^l, \ \mathbf{x}^u_{iter} \gets \mathbf{x}^u$
\STATE Solve:
\begin{equation*}
\begin{aligned}
x_{i}^{*l} := \underset{\mathbf{x},\mathbf{y}}{\min} \ x_{i}; \ \ &x_{i}^{*u} := \underset{\mathbf{x},\mathbf{y}}{\max} \ x_{i} \ \forall i=1,\ldots,n   & (a) \\
\text{subject to} \ \ &  f(\mathbf{x},\mathbf{y},\mathbf{z}) \leq f(\mathbf{x}^*_{loc},\mathbf{y}^*_{loc},\mathbf{z}^*_{loc}), & (b) \\
 & g(\mathbf{x},\mathbf{y},\mathbf{z}) \leq 0, & (c) \\
&  h(\mathbf{x},\mathbf{y},\mathbf{z}) = 0, & (d) \\
&  z_{K} = \langle x_ix_j\ldots x_k\rangle^{MC}, \ \ \forall K \in ML & (d) \\
&  \mathbf{x}^l_{iter} \leq \mathbf{x} \leq \mathbf{x}^u_{iter}, & (e) \\
&  \mathbf{y} \in \{0,1\}^m & (f) \\
\end{aligned}
\end{equation*}
\STATE $\mathbf{\mathbf{x}}^l \gets \mathbf{\mathbf{x}}^{*l}, \ \mathbf{\mathbf{x}}^u \gets \mathbf{\mathbf{x}}^{*u}$
\ENDWHILE
\STATE return $\mathbf{\mathbf{x}}^l,\mathbf{\mathbf{x}}^u$ (tightened bounds).
\end{algorithmic}	
\end{algorithm}

\subsection{Algorithm for global optimization of MINLPs}
\label{subsec:algo}
The second stage of CP-DTMC derives piecewise McCormick relaxations of multi-linear terms based on
multivariate dynamic partitioning.
In practice, partitioning the bounds of the variables of the McCormick tightens the overall relaxation. As the number of partitions goes to $\infty$, partitioning exactly approximates the original multi-linear terms.
However, introducing a large number of partitions generally renders the problem intractable because the choice of partition is controlled by binary on/off variables. Thus, typical approaches assume a (small) finite number of partitions that uniformly discretize the multi-linear variables
\cite{grossmann2013systematic,bergamini2008improved,castro2015tightening,hasan2010piecewise}. While this is a straight-forward method for partitioning the domain of variables, it potentially creates partitions that correspond to solutions that are far away from the optimality region of the search space. In other words, many of the partitions are not useful. Instead, we develop
an approach that successively tightens 
the McCormick relaxations with sparse domain discretization. This approach focuses partitioning on areas of the variable domain that appear to influence optimality the most.  

\noindent\textbf{Lower bounds using piecewise McCormick relaxations}
Without loss of generality and for ease of explanation, we restrict the discussion of the lower bounding procedure to bilinear terms\footnote{This approach is easily extended to multi-linear terms using successive bilinear relaxations as discussed in section \ref{subsec:relaxations}.}. 
Given a bilinear term $x_ix_j$, we partition the domains of $x_i$ and $x_j$ into $M_i \in \mathbb{Z}^+$ and $M_j \in \mathbb{Z}^+$ disjoint regions with new binary variables $\hat{\mathbf{y}}_i \in \{0,1\}^{M_i}$ and $\hat{\mathbf{y}}_j \in \{0,1\}^{M_j}$ added to the formulation. The binary variables are used to control the partitions that are active and the tighter relaxation associated with the active partition.
Formally, the piecewise McCormick constraints for a bilinear term, denoted by $\widehat{x_{ij}} \in \langle x_i,x_j\rangle^{UTMC}$ (uniform partitioning) or $\widehat{x_{ij}} \in \langle x_i,x_j\rangle^{DTMC}$ (dynamic partitioning), take the following form: 
\vspace{-0.12cm}
% {\fontsize{9}{8}\selectfont
\begin{subequations}
\begin{align}
   &\widehat{x_{ij}} \geq (\mathbf{x}_i^l\cdot\hat{\mathbf{y}}_i) x_j + (\mathbf{x}_j^l\cdot\hat{\mathbf{y}}_j) x_i - (\mathbf{x}_i^l\cdot\hat{\mathbf{y}}_i)(\mathbf{x}_j^l\cdot\hat{\mathbf{y}}_j)\\ 
   &\widehat{x_{ij}} \geq (\mathbf{x}_i^u\cdot\hat{\mathbf{y}}_i) x_j + (\mathbf{x}_j^u\cdot\hat{\mathbf{y}}_j) x_i - (\mathbf{x}_i^u\cdot\hat{\mathbf{y}}_i)(\mathbf{x}_j^u\cdot\hat{\mathbf{y}}_j)\\ 
   &\widehat{x_{ij}} \leq (\mathbf{x}_i^l\cdot\hat{\mathbf{y}}_i) x_j + (\mathbf{x}_j^u\cdot\hat{\mathbf{y}}_j) x_i - (\mathbf{x}_i^l\cdot\hat{\mathbf{y}}_i)(\mathbf{x}_j^u\cdot\hat{\mathbf{y}}_j)\\ 
   &\widehat{x_{ij}} \leq (\mathbf{x}_i^u\cdot\hat{\mathbf{y}}_i) x_j + (\mathbf{x}_j^l\cdot\hat{\mathbf{y}}_j) x_i - (\mathbf{x}_i^u\cdot\hat{\mathbf{y}}_i)(\mathbf{x}_j^l\cdot\hat{\mathbf{y}}_j) \\
   &\sum_{k=1}^{M_i} \hat{y_i}_k = 1, \ \ \sum_{k=1}^{M_j} \hat{y_j}_k = 1 \\
   &\hat{\mathbf{y}}_i \in \{0,1\}^{M_i}, \hat{\mathbf{y}}_j \in \{0,1\}^{M_j}
\end{align}
\label{eq:tmc}
\end{subequations}
% }
\vspace{-0.2cm}

\noindent where, $(\mathbf{x}_i^l, \mathbf{x}_i^u) \in \mathbb{R}^{M_i}$ are the lower and upper bound vectors of variable $x_i$ for each partition.
%and a similar definition holds true for $(\mathbf{x}_j^l, \mathbf{x}_j^u) \in \mathbb{R}^{M_j}$. 
In other words, for the $k^{th}$ partition of $x_i$, the following constraint defines the partition: ${x_i}^l_k \leq x_i \leq {x_i}^u_k$. Note that the bilinear terms in $\hat{\mathbf{y}}_j x_i$ and $\hat{\mathbf{y}}_i x_j$ are exactly linearized using standard McCormick relaxations. Also, $(\mathbf{x}_i^l\cdot\hat{\mathbf{y}}_i)(\mathbf{x}_j^l\cdot\hat{\mathbf{y}}_j)$ is rewritten as  $\mathbf{x}_i^l(\hat{\mathbf{y}}_i\hat{\mathbf{y}}_j^T)\mathbf{x}_j^l$, where $\hat{\mathbf{Y}}=(\hat{\mathbf{y}}_i\hat{\mathbf{y}}_j^T)$ is an $M_i\times M_j$ matrix with binary product entries. As discussed in section \ref{subsec:relaxations}, any binary product entry, $y_iy_j$, of $\hat{\mathbf{Y}}$ is exactly represented as $\langle y_i, y_j \rangle^{BMC}$. \\

\begin{figure}[htp]
   \centering
   \subfigure[Bilinear term $(x_ix_j)$]{
   \includegraphics[scale=0.34]{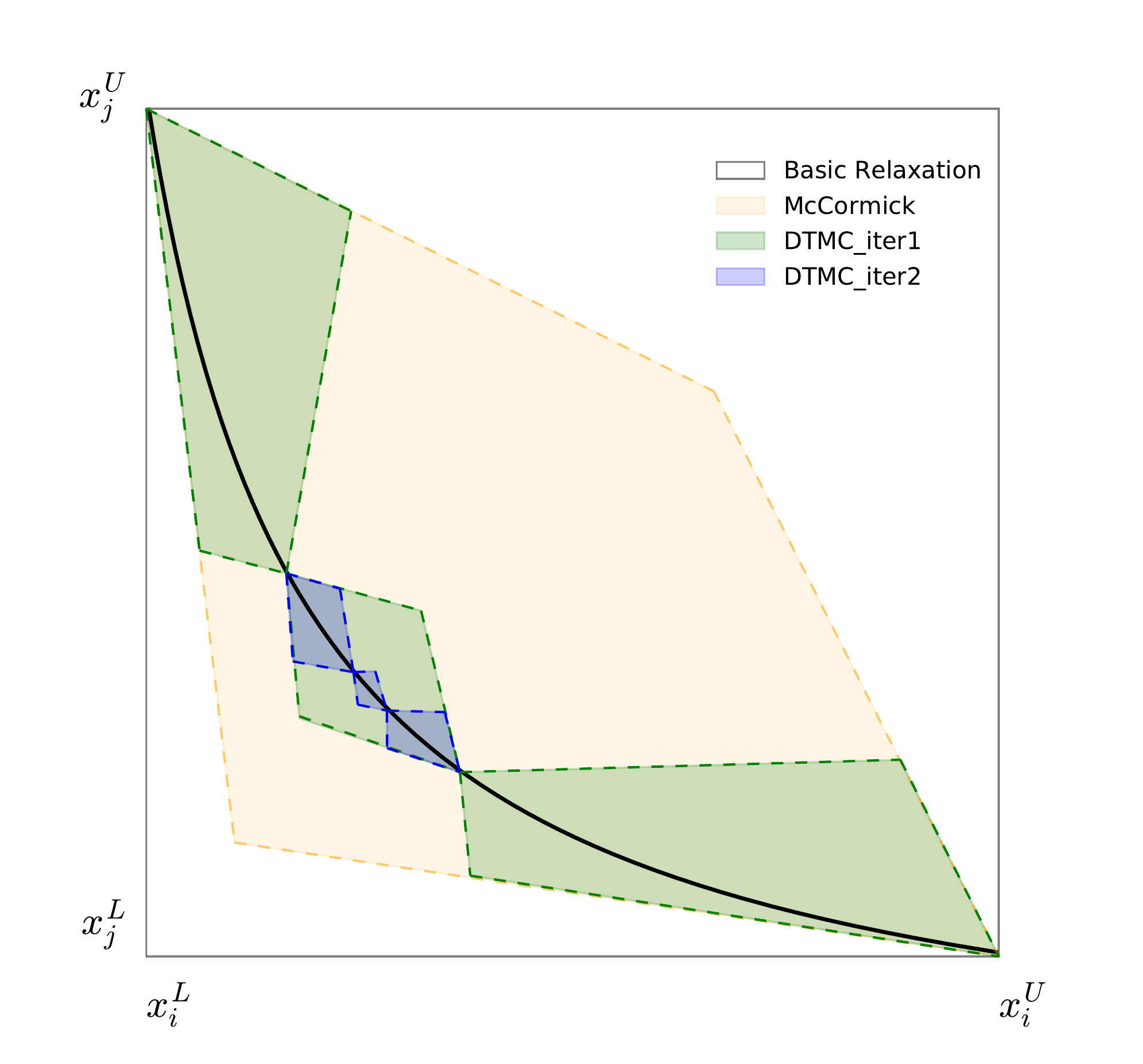}}
   \subfigure[Monomial term $(x_i^2)$]{
   \includegraphics[scale=1.3]{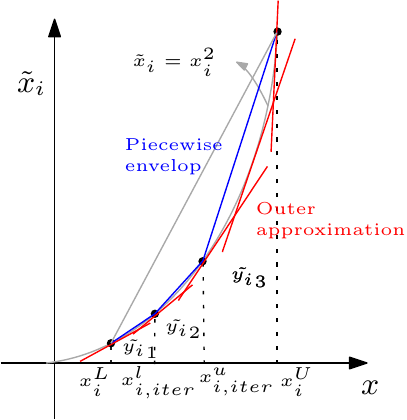}} 
   \caption{Feasible regions for bilinear and monomial (quadratic) terms based on DTMC.}
   \label{fig:mcc_region}
\end{figure}

\noindent \textbf{CP-DTMC algorithm for multi-linear terms} Given this model of piecewise McCormick relaxations, we can now formalize dynamically tightening of these relaxations. 
The pseudo-code of the DTMC algorithm is outlined in Algorithm \ref{Algo:partition}.  The full CP-DTMC algorithm combines Algorithm \ref{Algo:bound} with  Algorithm \ref{Algo:partition} and is described in Algorithm \ref{Algo:CPDTMC}.
We first discuss the dynamic partitioning scheme as outlined in Algorithm \ref{Algo:partition} followed by the discussion of Algorithm \ref{Algo:CPDTMC}. 

We first define $\boldsymbol{P}^*_{iter}$ as a vector of active partitions whose dimension is equal to $|\mathbf{x}|$. For any variable $x_i$, an active partition contains a lower bound and an upper bound for $x_i$
The choice of the active partition of $x_i$ is the binary variable of vector $\hat{\mathbf{y}}_i$ whose component is equal to 1.0. As shown in line 3 of  Algorithm \ref{Algo:partition}, the size of the partition  
is dependent on the size of the active partition of the current solution $\mathbf{x}^*_{iter}$. 
The parameter, $\Delta$, is used to scale the partition's size and it influences the convergence speed and the number of partitions.  
Lines 4-10 ensure that the partition's size is greater than a prescribed tolerance and that the partition lies within the contracted bounds. 
\begin{figure}[htp]
    \centering
    \includegraphics[scale=0.98]{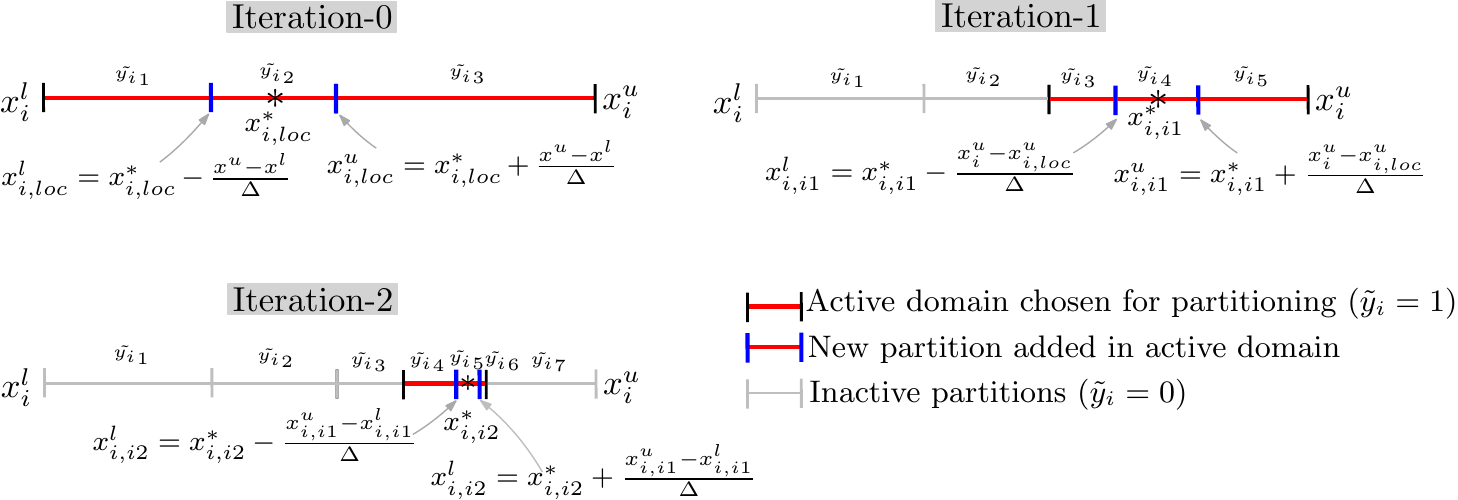}
    \caption{Dynamic partitioning of
    variable $x_i$ as described in Algorithm \ref{Algo:partition}}
    \label{fig:partitions}
\end{figure}

\noindent 
In Algorithm \ref{Algo:CPDTMC}, lines 1-3 execute Algorithm \ref{Algo:bound} to tighten the bounds using the feasible solution $(\mathbf{x}^*_{loc},\mathbf{y}^*_{loc},\mathbf{z}^*_{loc})$. Interestingly, on some MINLPs, this process shrank the gap between the upper and lower bounds on some variables to 0. Line 5 initializes the tightened bound domains as the active partitions as illustrated in ``iteration-0" of Figure \ref{fig:partitions}. Lines 6-12 
iteratively add dynamic partitions around the current solution $\mathbf{x}^*_{iter}$. Iterations 1 and 2 of Figure \ref{fig:partitions} clearly illustrate the partitioning scheme employed in this algorithm. 
%represent an iterative lower bounding procedure where the partitions are dynamically added around the current solution $\mathbf{x}^*_{iter}$. 
The iterations stop (line 6) when a) the normalized improvement of the lower bound is less than $TOL_{imp}$, b) $\mathbf{x}^*_{iter}$ remains in the same partitions and the size of the partitions is $\le \boldsymbol{\epsilon}$, or c) the computation hits a time limit. In Figure \ref{fig:partitions}, the third iteration terminates if the $x^*_{i,i3}$ remains in partition $[x^l_{i,i2},x^u_{i,i2}]$ and its size is less than $\epsilon_i$. Figure \ref{fig:mcc_region}(a) is a geometric example of a DTMC iteration applied to a bilinear term. This figure illustrates how the area enclosed by the convex relaxations decreases as partitions are applied.

%----------------------------------------------;
%  Algorithm: Dynamic partitioning algorithm   ;
%----------------------------------------------;
\begin{algorithm}[htp]
\footnotesize
\caption{Dynamic partitioning of variable domains}
Notation: Let $\boldsymbol{P}^*_{iter}$ represent a vector of active partitions for variable vector $\mathbf{x}$. $\mathbf{x}^l({\boldsymbol{ P}^*_{iter}})$ and $\mathbf{x}^u({\boldsymbol{ P}^*_{iter}})$ represent the vectors of lower and upper bounds of the active partitions of  $\mathbf{x}$ respectively.
\label{Algo:partition}
\begin{algorithmic}[1]
\STATE Input: $\mathbf{x}^l,\mathbf{x}^u, \mathbf{x}^*_{iter}$, $\boldsymbol{P}^*_{iter}$, $\boldsymbol{\epsilon} > 0$, $\boldsymbol{P}^*_{new} = \emptyset$, $\Delta > 0$
\STATE $\mathbf{lb} \gets \mathbf{x}^l({\boldsymbol{ P}^*_{iter}})$, $\mathbf{ub} \gets \mathbf{x}^u({\boldsymbol{ P}^*_{iter}})$
\STATE Evaluate the size of new partition $$\mathbf{l}_{iter} = \frac{\mathbf{ub} - \mathbf{lb}}{\Delta}$$
\IF{$\mathbf{l}_{iter} > \boldsymbol{\epsilon}$}
\STATE $\mathbf{v}^l \gets \max(\mathbf{x}^l,(\mathbf{x}^{iter} - \mathbf{l}_{iter})), \mathbf{v}^u \gets \min(\mathbf{x}^u,(\mathbf{x}^{iter} + \mathbf{l}_{iter}))$
\STATE $\boldsymbol{P}^*_{new} \gets \{ (v^l_i,v^u_i), \ \forall i=1,\ldots,n \}$
\STATE return $\boldsymbol{P}^*_{new}$
\ELSE
\STATE return $\emptyset$
\ENDIF
\end{algorithmic}	
\end{algorithm}

%----------------------------------------------;
%  Algorithm: CP-DTMC algorithm     ;
%----------------------------------------------;
\begin{algorithm}[htp]
\footnotesize
%\caption{\textbf{}}
\caption{An algorithm for global optimization of MINLPs (CP-DTMC)}
\label{Algo:CPDTMC}
\begin{algorithmic}[1]
\STATE Input: MINLP, $TOL_{imp}>0$
\STATE Obtain local solution $(\mathbf{x}^*_{loc},\mathbf{y}^*_{loc},\mathbf{z}^*_{loc})$ for the given MINLP
\STATE Execute Algorithm \ref{Algo:bound} $(\mathbf{x}^*_{loc},\mathbf{y}^*_{loc},\mathbf{z}^*_{loc})$ to calculate bounds $(\mathbf{x}^l,\mathbf{x}^u)$ on variables $\mathbf{x}\in \mathbb{R}^n$ appearing in multi-linear terms.
%\STATE Based on sequential-CP algorithm described in \ref{Algo:bound}, obtain the tightest possible bounds $(\mathbf{x}^l,\mathbf{x}^u)$ on variables $\mathbf{x}\in \mathbb{R}^n$ appearing in multi-linear terms.
\STATE $\mathbf{x}^*_{iter} \gets \mathbf{x}^*_{loc},\mathbf{y}^*_{iter} \gets \mathbf{y}^*_{loc}$
\STATE $\boldsymbol{P}^*_{iter} \gets \{(x_i^l,x_i^u), \ \forall  i=1,\ldots, n\}$ (Initialize the active partitions with the entire domains of variables)
\WHILE{Stopping criterion not satisfied}
\STATE  For the current $\mathbf{x}^*_{iter}$ and $\boldsymbol{P}^*_{iter}$, obtain $\boldsymbol{P}^*_{new}$ from Algorithm \ref{Algo:partition}. 
\STATE $\boldsymbol{P}^*_{iter} \gets (\boldsymbol{P}^*_{iter} \cup \boldsymbol{P}^*_{new})$ (updated partitions for DTMC in line 9)
\STATE Solve
\begin{equation*}
\begin{aligned}
\calP_{iter}: \ \ \ & \underset{\mathbf{x},\mathbf{y}}{\text{minimize}} & &  f(\mathbf{x},\mathbf{y},\mathbf{z}) \\
&\text{subject to} & & g(\mathbf{x},\mathbf{y},\mathbf{z}) \leq 0, \\
& & & h(\mathbf{x},\mathbf{y},\mathbf{z}) = 0, \\
& & & z_{K} = \langle x_ix_j\ldots x_k\rangle^{DTMC}, \ \ \forall K \in ML\\
& & & \mathbf{x}^l \leq \mathbf{x} \leq \mathbf{x}^u, \\
& & & \mathbf{y} \in \{0,1\}^m
\end{aligned}
\end{equation*}
\STATE Let $(\mathbf{x}^*_{iter},\mathbf{y}^*_{iter})$ be the solution to $\calP_{iter}$
\STATE Update the vector of active partition sets $\boldsymbol{P}^*_{iter}$ such that the binary variable $\hat{y}^*_i$ on $x_i$ is equal to 1.0.
\ENDWHILE
\STATE Output: Global optimum solution $(\mathbf{x}^*_{opt},\mathbf{y}^*_{opt})$ or a lower bound (if solver times out) to the MINLP. 
\end{algorithmic}	
\end{algorithm}

\noindent \textbf{CP-DTMC Generalization} It is important to note that this approach can be applied to other types of relaxations. For example, consider monomials whose powers contain positive integer exponents $(\geq 2)$.
Without loss of generality\footnote{In the case of higher order monomials, i.e., $x_i^5$, we apply a reduction of the form $x_i^2x_i^2x_i \Rightarrow \tilde{x_i}^2x_i \Rightarrow \tilde{\tilde{x_i}}x_i$.}, we assume the monomial takes the form $x_i^2$.
We once again partition the domain of $x_i$ into $N_i \in \mathbb{Z}^+$ disjoint regions. Let $\tilde{\mathbf{y}_i} \in \{0,1\}^{N_i}$ be the binary variables added to the formulation. Formally, this piecewise convex relaxation, denoted by $\tilde{x_i} \in \langle x_i\rangle^{DTMC-q}$, takes the form: 
\vspace{-0.15cm}
% {\fontsize{9}{8}\selectfont
\begin{subequations}
\begin{align}
   &\tilde{x_i} \geq x_i^2, \\ 
   &\tilde{x_i} \leq  \left((\mathbf{x}_i^l\cdot\tilde{\mathbf{y}}_i) + (\mathbf{x}_i^u\cdot\tilde{\mathbf{y}}_i)\right) x_i - (\mathbf{x}_i^l\cdot\tilde{\mathbf{y}}_i)(\mathbf{x}_i^u\cdot\tilde{\mathbf{y}}_i)\\ 
   &\sum_{k=1}^{N_i} \tilde{y_i}_k = 1 \\
   &\tilde{\mathbf{y}_i} \in \{0,1\}^{N_i}
\end{align}
\label{eq:qc}
\end{subequations}
% }

\noindent
Note that $(\mathbf{x}_i^l\cdot\tilde{\mathbf{y}}_i)(\mathbf{x}_i^u\cdot\tilde{\mathbf{y}}_i)$ can be rewritten as  $\mathbf{x}_i^l(\tilde{\mathbf{y}}_i\tilde{\mathbf{y}}_i^T)\mathbf{x}_i^u$, where $\tilde{\mathbf{Y}}=(\tilde{\mathbf{y}}_i\tilde{\mathbf{y}}_i^T)$ is an $N_i\times N_i$ symmetric matrix with binary product entries (squared binaries on diagonal). Hence it is sufficient to linearize the entries of the upper triangular matrix with exact representations as discussed in section \ref{subsec:relaxations}. This relaxation is then directly introduced into Algorithm \ref{Algo:CPDTMC}. The only modification is to supplement the convex envelops in $\calP_{iter}$ with these monomial terms.
\begin{lem}
Given a finite number of partitions on $x_i$, the piecewise convex relaxation of $\langle x_i\rangle^{DTMC-q}$ is strictly tighter than $\langle x_i,x_i\rangle^{DTMC}$.
\end{lem}
\begin{proof}
For a given, finite number of partitions, $N_i$, on variable $x_i$, $\langle x_i,x_i\rangle^{DTMC}$ reduces to the following three-inequalities representing the piecewise convex relaxations:
\begin{subequations}
\begin{align}
\label{eq:pf1}
    &\tilde{x_i} \geq  2(\mathbf{x}_i^l\cdot\tilde{\mathbf{y}}_i) x_i - (\mathbf{x}_i^l\cdot\tilde{\mathbf{y}}_i)^2\\
    \label{eq:pf2}
    &\tilde{x_i} \geq  2(\mathbf{x}_i^u\cdot\tilde{\mathbf{y}}_i) x_i - (\mathbf{x}_i^u\cdot\tilde{\mathbf{y}}_i)^2\\ \label{eq:pf3}
   &\tilde{x_i} \leq  \left((\mathbf{x}_i^l\cdot\tilde{\mathbf{y}}_i) + (\mathbf{x}_i^u\cdot\tilde{\mathbf{y}}_i)\right) x_i - (\mathbf{x}_i^l\cdot\tilde{\mathbf{y}}_i)(\mathbf{x}_i^u\cdot\tilde{\mathbf{y}}_i) \\ 
   &\sum_{k=1}^{N_i} \tilde{y_i}_k = 1, \ \tilde{\mathbf{y}_i} \in \{0,1\}^{N_i} \nonumber
\end{align}
\label{eq:dtmc_proof}
\end{subequations}

\noindent
Clearly, inequalities \eqref{eq:pf1} and \eqref{eq:pf2} are under estimators of $x_i^2$ at grid points $x_i^l,i=1,\ldots,N_i$ and $x_{N_i}^u$ respectively. The over estimator in \eqref{eq:pf3} is same as the over estimator defining $\langle x_i\rangle^{DTMC-q}$. Further, the second-order conic under estimator of $\langle x_i\rangle^{DTMC-q}$ can be equivalently represented with infinitely many linear inequalities. However, as discussed above, the under estimators in $\langle x_i,x_i\rangle^{DTMC}$ are finite ($N_i+1$), thus relaxing the second-order-cone. Therefore, $\langle x_i\rangle^{DTMC-q} \subset \langle x_i,x_i\rangle^{DTMC}$. 
\end{proof}

Because of Lemma 31, we use this relaxation rather than McCormick on monomial terms.
However, using this relaxation forced us to introduce a technical subtlety into the algorithm implementation.
While constraint $\tilde{x_i} \geq x_i^2$ in \eqref{eq:qc} is a convex, second order cone (SOC), 
several moderately sized problems were difficult to solve, even with modern,
state-of-the-art solvers (CPLEX). Either the solver convergence was very slow or they terminated with a numerical error. To circumvent this issue, we implemented a cutting-plane approach for these constraints. This approach relaxes the SOC constraint with a finite number of valid cutting planes (first order derivatives), produces an outer envelop, and produces a lower bound on the optimal solution. This lower bound is tightened for every violated SOC constraint by adding the corresponding valid cutting plane until a solution obtained is feasible, and hence optimal, for the original SOC set. Figure \ref{fig:mcc_region}(b) illustrates the outer-approximation procedure. Red colored lines are the under estimators of $x_1^2$ and the valid cutting planes added to the formulation. In Algorithm \ref{Algo:CPDTMC}, this approach is used for the solve routine of line 9. We expect the need for this technical detail to diminish as conic solvers improve.

\noindent \textbf{TCP-DTMC - A hybrid approach} 
The main idea behind the TCP-DTMC approach is to combine the sequential bound tightening procedure in Algorithm \ref{Algo:bound} with a three-partition piecewise McCormick relaxation on every variable in multi-linear terms.
%instead of the standard McCormick relaxation. 
Since we know $\mathbf{x}^*_{loc}$ from a local solver, we discretize the domain with atmost three partitions and satisfy the rules of partitioning as described in Algorithm \ref{Algo:partition}. Therefore, in line 4(d) of Algorithm \ref{Algo:bound}, the McCormick relaxations are replaced by 
$$z_{K} = \langle x_ix_j\ldots x_k\rangle^{DTMC}, \ \ \forall K \in ML.$$
with an additional constraint, 
$$\sum_{k=1}^{3} \tilde{y_i}_k = 1 \ \forall i=1,\ldots,|\mathbf{x}|.$$
The primary intuition behind this bound tightening procedure is to obtain tighter bounds around the local solution and possibly converge the bounds to near-optimum solutions in the initial step.

%!TEX root = ../tmc.tex
% Numerical Results

\section{Computational results}
\label{sec:results}

\vspace{-.3cm}

All computations were performed using the high performance computing resources at Los Alamos National Laboratory (using nodes for parallel computation) with Intel(R) Xeon(R) CPU E5-2660 v3 @ 2.60GHz processors and 62GB of memory. All MILPs were solved using CPLEX 12.6.2 with default options and presolver switched on. All the outer-approximation cutting planes for quadratic terms were implemented as a CPLEX lazy cut callback.
BARON 15.2.0 (default options) was the global solver used to benchmark the performances of CP-DTMC and TCP-DTMC.  
Ipopt 3.12.4 and Bonmin 1.8.4 were used as the local NLP and MINLP solvers, respectively. These solvers were used to produce the initial feasible solution for Algorithm \ref{Algo:bound}. Table \ref{tab:parameters} summarizes the values of all the parameters used in CP-DTMC. The notation ``TO" is used to indicate when the algorithm timed out (time limit=3600 sec) and ``GOpt"
is used to indicate global optimum, i.e. the lower bound is within 0.0001\% of the known optimal solution. In Table \ref{tab:results}, Best$\Delta$ and Best$N$ correspond to the best solution found within the CPU limit for DTMC's $\Delta$ and UTMC's number of partitions, respectively. Also, in Table \ref{tab:results}, we define the following: 

{\fontsize{9}{8}\selectfont
\begin{align*}
    \mathrm{\% Gap} = \frac{f(\mathbf{x}_{opt}^*,\mathbf{y}_{opt}^*,\mathbf{z}_{opt}^*) - f(\mathbf{x}_{iter}^*,\mathbf{y}_{iter}^*,\mathbf{z}_{iter}^*)}{f(\mathbf{x}_{iter}^*,\mathbf{y}_{iter}^*,\mathbf{z}_{iter}^*)}\times 100, \ \  
    \mathrm{\% BC} = \frac{\|\mathbf{x}^U-\mathbf{x}^L\|_2 - \|\mathbf{x}^u-\mathbf{x}^l\|_2}{\|\mathbf{x}^u-\mathbf{x}^l\|_2}\times 100
 \end{align*}}

\noindent 
In our numerical experiments we considered \textit{three} NLPs and \textit{thirteen} MINLPs that ranged from small, contrived examples to large-scale MINLP benchmark problems selected from MINLPLib 2 \cite{bussieck2003minlplib}. 
We chose problems whose nonlinearity is expressed with multi-linear terms. The MINLPs chosen for analysis purposes are not exhaustive and we will expand the test-bed in our future work. Table \ref{tab:data} summarizes the statistics of the test-bed including global optimum, number of constraints, binary variables, continuous variables and multi-linear terms. Note that ``nlp2" contains two, fourth degree monomial terms and ``eniplac" contains bilinear, quadratic and cubic monomials. In the case of the ``blend" instances, we partition only a single variable per bilinear term as these were large scale MINLPs\footnote{In the ``blend" instances, there were few binary variables that appeared in most of the bilinear terms. These are the variables chosen for partitioning}.

\vspace{-0.5cm}
\begin{table}[H]
    \centering
    \footnotesize
    \caption{Parameters used in CP-DTMC}
    \begin{tabular}{c|c}
    \toprule
         $N$ (number of partitions in UTMC) & 10, 20, 40 \\
         $\Delta$ (scaling parameter in DTMC/TCP) & 2, 4, 8, 10, 16, 32 \\
         Wall time execution limit & 3600.0 sec\\
         $\boldsymbol{\epsilon}$  (minimum partition length tolerance) & 0.001 \\
        %  Lazy cuts tolerance & 0.00001 \\
         $TOL$ (bound tightening tolerance) & 0.01 \\
         $TOL_{imp}$ (\% improvement tolerance in DTMC) & 0.001\% \\
        %  $TOL_{opt}$ () & 0.0001\% \\
    \bottomrule 
    \end{tabular}
    \label{tab:parameters}
\end{table}

\vspace{-1.2cm}
\begin{table}[H]
\caption{Problem Description}
\scriptsize
% \tiny
\centering
{%
\begin{tabular}{lccccc}
\toprule
Instance& GOpt& \#Constraints &\#BVars & \#CVars& \#ML \\
&&&& (\#CVars-discretized)&   \\
\hline
nlp1&	58.384&	3& 0&	2(2)&	3\\
nlp2&	0& 2&	0&	2(4)&	4 \\
nlp3&	7049.248& 14& 0&	8(8)&	5 \\
ex1223a&	4.580& 9&	4&	3(3)&	3 \\
ex1264&	8.6& 55&	68&	20(20)&	16 \\
ex1265&	10.3&  74& 100&	30(30)&	25 \\
ex1266&	16.3& 95& 138&	42(42)&	36 \\
fuel& 8566.119& 15& 3& 12(6)&	3 \\
meanvarx& 14.369& 44& 14& 21(7)& 28 \\
util& 4.305& 167&  28& 117(7)& 5 \\
eniplac& -132117.083& 189& 24& 117(24)& 66 \\
blend029& 13.359& 213& 36&	66(10)&	28\\ 
blend531& 20.039& 736& 104&	168(28)& 146 \\
blend718& 7.394& 606& 87& 135(20)& 100 \\
blend480& 9.227& 884& 124& 188(28)& 152 \\
blend146& 45.297& 624& 87& 135(20)& 104 \\
\bottomrule 
\end{tabular}}
\label{tab:data}
\end{table}%

\begin{figure}[h]
   \centering
   \subfigure[Mathematical formulations]{
   \includegraphics[scale=0.67]{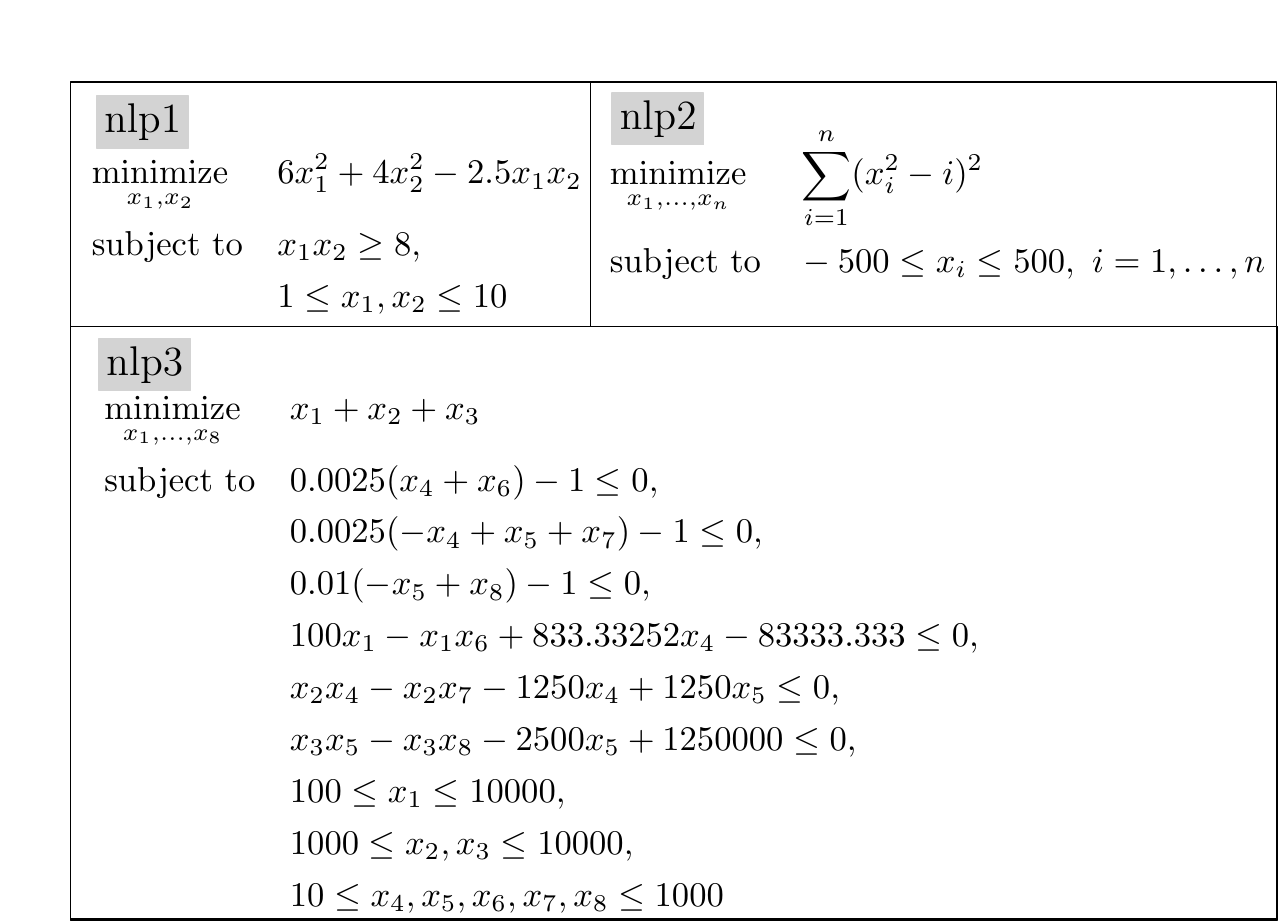}}
   \subfigure[nlp2 with multiple global minima and a local minimum]{
   \includegraphics[scale=0.27]{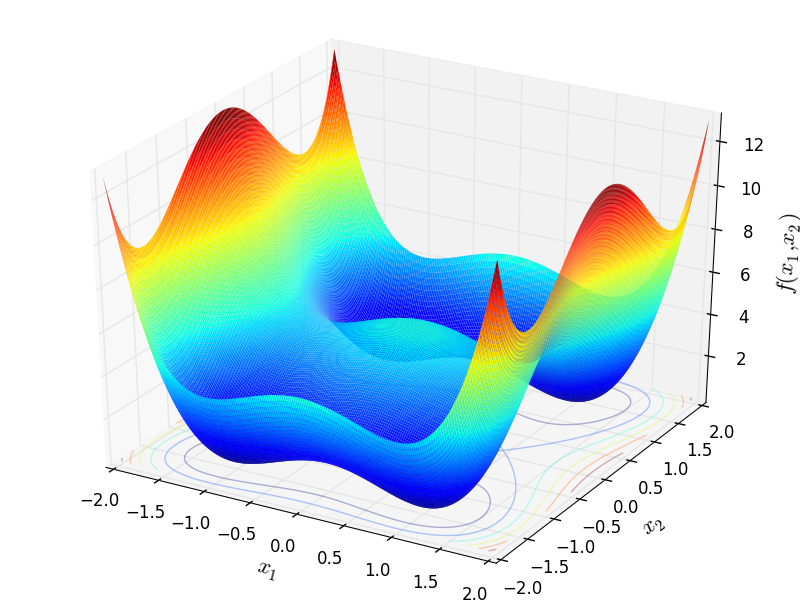}} 
   \caption{NLPs considered in this paper}
   \label{Fig:NLPs}
\end{figure}

\vspace{-0.5cm}

\subsection{NLPs}
We first consider a small set of simple NLPs, as described in Figure \ref{Fig:NLPs}(a) and \cite{castro2015tightening,kolodziej2013discretization,teles2013univariate}. We compare the performance of our algorithms with  BARON. These problems are interesting to discuss in more detail.
``nlp1", taken from \cite{castro2015tightening}, involves both bilinear and quadratic terms. ``nlp2" appears in applications related to electromagnetic inverse scattering problems \cite{jamil2013literature}. In this problem, quadrilinear terms in the objective and large bounds on the variables makes it particularly challenging for existing McCormick-relaxation based algorithms. For computational studies, we solve nlp2 in two dimensions $(n=2)$. As shown in Figure \ref{Fig:NLPs}, the objective function has multiple global minima at $\Mypm(1,\sqrt{2})$ and a local minimum at the origin. When solved with IPOPT we get
a local solution, $f^*_{loc}= 5$, at $(0,0)$. ``nlp3", taken from a standard test-suite \cite{hock1980test}, has five bilinear terms and large bounds on all the variables. Since this is a challenging problem for the equally partitioned, piecewise McCormick relaxations, this problem has been studied in detail in \cite{castro2015tightening,castro2015normalized,teles2013univariate}. 

\noindent \textbf{Computational Performance}
Table \ref{tab:results} summarizes the performance of the algorithms on the NLPs. On nlp1 and nlp2, the  algorithms performed consistently better than Baron. For nlp2, we observed that the quadratic convex envelopes, in conjunction with outer-approximation, performed computationally better than solving mixed-integer SOCs.   

%\vspace{-0.5cm}
\begin{table}[htp]
\caption{Contracted bounds after applying sequential tightened-CP algorithm on nlp3.}
\scriptsize
\centering
% \resizebox{0.5\textwidth}{!}
{%
\begin{tabular}{cccccccc}
\toprule
 & \multicolumn{2}{c}{Original bounds} & \multicolumn{2}{c}{TCP bounds} &
 \multicolumn{3}{c}{\#BVars added}\\
 \cmidrule(lr){2-3}
 \cmidrule(lr){4-5}
 \cmidrule(lr){6-8}
 Variable & $L$ & $U$ & $l$ & $u$ & DTMC & CP-DTMC & TCP-DTMC \\
 &&&&&($\Delta=4$)&($\Delta=10$)&($\Delta=10$) \\
 \midrule
$x_1$ &100 &10000 &573.1 &585.1 &14 &14 &3+3 \\
$x_2$ &1000 &10000 &1351.2 &1368.5 &14 &14 &3+3\\
$x_3$ &1000 &10000 &5102.1 &5117.5 &17 &15 &3+3\\
$x_4$ &10 &1000 &181.5 &182.5 &16 &15 &3+3\\
$x_5$ &10 &1000 &295.3 &296.0 &17 &15 &3+3\\
$x_6$ &10 &1000 &217.5 &218.5 &16 &15 &3+3\\
$x_7$ &10 &1000 &286.0 &286.9 &17 &15 &3+3\\
$x_8$ &10 &1000 &395.3 &396.0 &17 &15 &3+3\\
\cmidrule(r){1-8}
Total &&&&&128&118&48 \\
\bottomrule
\end{tabular}}
\label{tab:tighten}
\end{table}

We performed a detailed study of nlp3 as this problem has received considerable interest in the literature. Table \ref{tab:tighten} show the effectiveness of sequential tightened-CP (TCP) techniques when applied to nlp3. The initial large global bounds are tightened by {\it a few orders of magnitude} with the addition of three binary variables per continuous variable in the bilinear terms. This shows the value of combining the disjunctive polyhedral approximation around the initial feasible solution ($\mathbf{x}^*_{loc}$) with the bound tightening procedure. In Figure \ref{Fig:p3_BL} we also observe that the additional variables do not increase the overall run time too much.
More importantly, the reduction in the variable domains is \textit{substantial} using TCP. Finally, the jump in the run time after the first iteration in Figure \ref{Fig:p3_BL}(b) is due to the 
reduction in the initial bounds using the CP/TCP algorithm. 

\noindent \textbf{Parameter tuning} Table \ref{tab:delta} shows the performance of the algorithm on nlp3 for varying values of $\Delta$. It is clear that the solution time and the number of binary variables added in the DTMC algorithm depend on tuning this parameter. However, we note that the \% gap for all $\Delta \geq 4$ using TCP-DTMC were close to the optimal solution. For $\Delta=10$, the global optimum is found in 60 seconds with only 48 binary variables added to the formulation. Overall, for nlp3, it is important to note that the TCP-DTMC algorithm outperforms most of the state-of-the-art piecewise relaxation methods developed in the literature.

\begin{figure}[htp]
  \centering
  \subfigure[Tightened bounds after each iteration]{
  \includegraphics[scale=0.27]{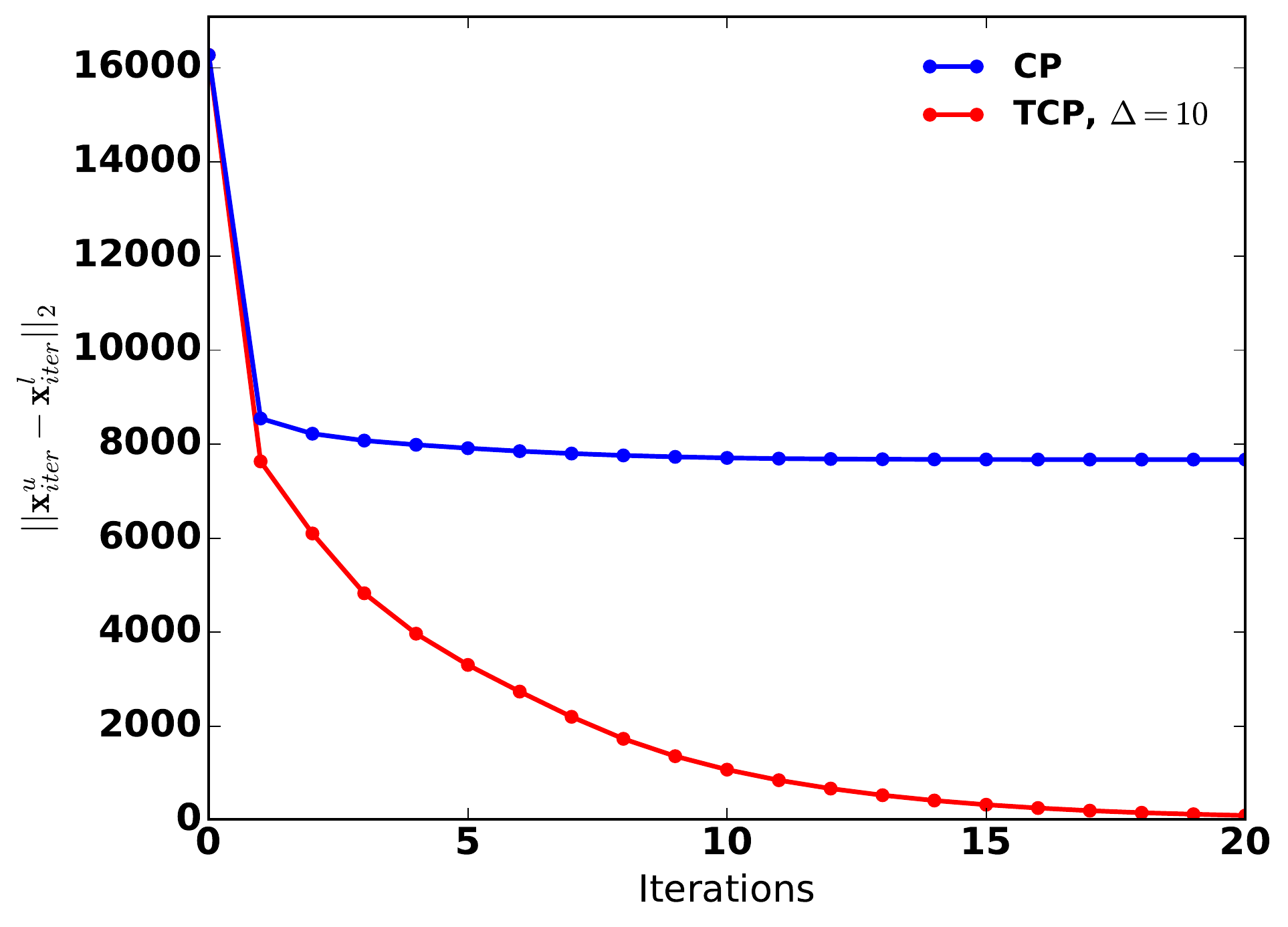}}       
  \subfigure[Elapsed time(sec) of bounds tightening]{
  \includegraphics[scale=0.27]{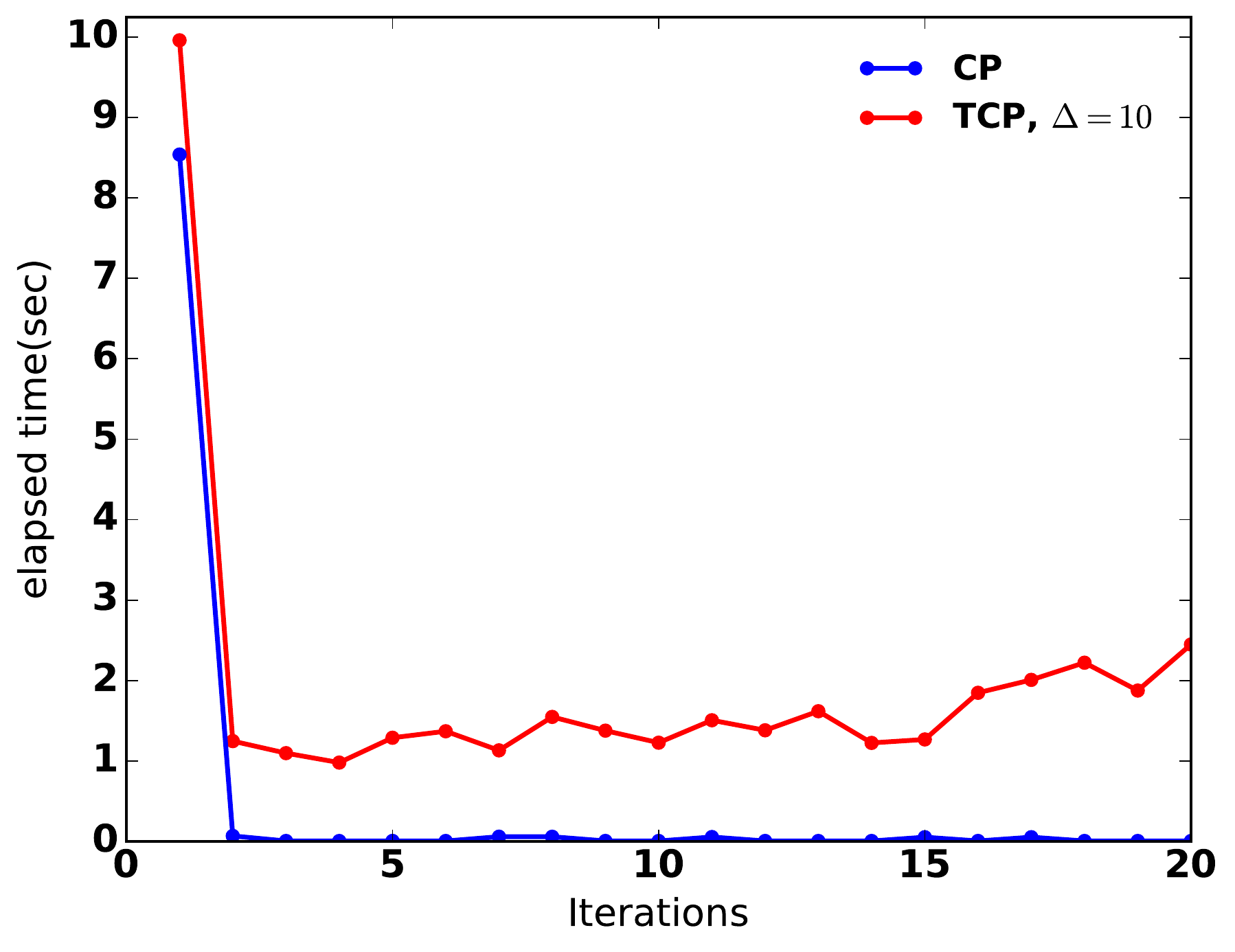}} 
  \caption{Performance of sequential CP and sequential tightened-CP on nlp3.}
  \label{Fig:p3_BL}
\end{figure}

%\{-1.3cm}

\begin{table}[htp]
\caption{Performance of proposed algorithms on nlp3 for various $\Delta$ values.}
\scriptsize
\centering
% \resizebox{0.5\textwidth}{!}
{%
\begin{tabular}{cccccccccc}
\toprule
 & \multicolumn{3}{c}{DTMC} & \multicolumn{3}{c}{CP-DTMC} &
 \multicolumn{3}{c}{TCP-DTMC}\\
 \cmidrule(lr){2-4}
 \cmidrule(lr){5-7}
 \cmidrule(lr){8-10}
 $\Delta$ & \#BVars & T & \%Gap & \#BVars & T & \%Gap & \#BVars & T & \%Gap \\
 \midrule
 
2& 137& 92.91& 141.14& 116& 1393.31&	39.654&	160& TO& 5.148 \\
4& 128& TO&	0.013& 114&	TO&	0.032& 48& 44.44& 0.00064 \\
8& 116& TO&	0.065& 117&	TO&	0.009& 48& 50.31& 0.00014 \\
10&	118& TO& 0.052&	118& TO& 0.004&	48&	59.50&	GOpt \\
16&	117& TO& 0.092&	119& TO& 0.009&	48&	63.04&	0.00027 \\
32&	118& TO& 0.076&	120& TO& 0.03& 48&	90.09& 0.00029 \\

\bottomrule
\end{tabular}}
\label{tab:delta}
\end{table}

%\vspace{-1cm}

\subsection{MINLPs}
In this section we compare the algorithms on MINLP benchmark problems described in Table \ref{tab:results}. 

\noindent \textbf{Performance of DTMC without CP/TCP}
From Table \ref{tab:results} it is apparent that dynamically partitioning variable domains to tighten McCormick relaxations is efficient even without bound tightening (CP/TCP). DTMC outperformed the uniform partitioning approach (UTMC) in twelve out of thirteen problems. Problems ex1266, meanvarx and blend718 
show the biggest performance gains (UTMC even times out on meanvarx).
DTMC also outperforms Baron on ten out of thirteen MINLPs, in particular on eniplac, blend531 and blend718. Blend718 is noteworthy as Baron times out with a 27.5\% optimality gap but DTMC produces global optimum solution within 326.2 sec. 

\noindent \textbf{Performance of DTMC with CP/TCP}
In Table \ref{tab:results}, we observed a reduction in run times of the DTMC algorithm (T$_{DTMC}$) due to CP/TCP bound tightening (with few exceptions). The reductions are significant on the large-scale blend480 and blend518 problems. Specifically, after TCP, DTMC performs almost twice as fast as Baron on blend480. 
It is also noteworthy to compare the performance of CP-DTMC and TCP-DTMC with Baron on these problems. We observed that Baron timed out on blend718 and blend146 with 27.5\% and 4.039\% optimality gaps. However, for blend718, CP-DTMC and TCP-DTMC produce global optimum solutions within 488 sec and 208 sec, respectively. On blend146, CP-DTMC and TCP-DTMC timed out with smaller optimality gaps (0.043\% and 0.0570\%) than Baron and UTMC. On blend531, Baron finds the global optimum in 2349 sec, but CP-DTMC and TCP-DTMC find the global optimum in 157 sec and 392 sec - at least fifteen times faster. However, on blend480, while the performance of our algorithms was better than UTMC, it was not better than Baron. 

%\vspace{-0.5cm}
\begin{table}[htp]
\caption{Comparison of all algorithms}
%\footnotesize
% \tiny
\centering
%\resizebox{1.0\textwidth}{!}
{%
\begin{tabular}{lrrrrrrrrrr}
\specialrule{1.1pt}{1pt}{1pt}
\hline
 &&\multicolumn{2}{c}{BARON} &\multicolumn{3}{c}{UTMC} &\multicolumn{3}{c}{DTMC} \\
 \cmidrule(lr){3-4}
 \cmidrule(lr){5-7}
 \cmidrule(lr){8-10}
 Instance && \%Gap  & T  & Best$N$  &  \%Gap  & T$_{UTMC}$  & Best$\Delta$  &  \%Gap  & T$_{DTMC}$ \\
 \midrule
nlp1 & & GOpt & 4.42 & 40 & 0.091 & 12.74 & 32 & GOpt & 1.71 \\ 
nlp2&& GOpt&	4.19&	20&	GOpt&	0.07&	32&	GOpt&	0.07 \\
nlp3&& GOpt&	13.26&	40&	0.585&	TO&	4&	0.013&	TO \\
ex1223a&& GOpt&	4.26&	20&	GOpt&	0.02&	32&	GOpt&	0.01 \\
ex1264&&	GOpt&	13.84&	10&	GOpt&	50.62&	10&	GOpt&	1.97&	\\
ex1265&&	GOpt&	7.93&	10&	GOpt&	76.35&	8&	GOpt&	0.57& \\
ex1266&&	GOpt&	17.43&	10&	GOpt&	114.15&	2&	GOpt&	0.74\\
fuel&& GOpt&	4.38&	40&	GOpt&	1.09&	32&	GOpt&	0.40 \\
meanvarx&&	GOpt&	4.31&	40&	0.221&	TO&	8&	0.012&	90.64\\
util&& GOpt&	5.54&	40&	8.186&	6.94&	32&	0.0098&	8.21 \\
eniplac&& GOpt&	330.46&	10&	GOpt&	2.47&	32&	GOpt&	1.97 \\
blend029&& GOpt&	15.33&	10&	GOpt&	2.51&	32&	GOpt&	1.95 \\
blend531&& GOpt&	2348.08& 20& 0.045& 153.43&	8&	GOpt&	140.76\\
blend718&& 27.484& TO& 20&	GOpt& 1198.42&	16&	GOpt&	326.17 \\
blend480&& GOpt&	2044.22&	20&	0.2&	TO&	16&	0.125&	2478.27 \\
blend146&& 4.039& TO& 20& 0.58& TO& 16& 0.035& TO \\
\specialrule{0.9pt}{1pt}{1pt}
&\multicolumn{5}{c}{CP-DTMC} &\multicolumn{5}{c}{TCP-DTMC}\\
\cmidrule(lr){2-6}
\cmidrule(lr){7-11}
Instance & Best$\Delta$  & BC(\%)  & \%Gap  &T$_{CP}$ & T$_{DTMC}$ & Best$\Delta$  & BC(\%)  & \%Gap  &T$_{TCP}$ & T$_{DTMC}$\\
\midrule

nlp1 & 16 & 96.67 & GOpt & 8.96 & 1.18 & 16 & 98.89 & GOpt & 2.73 & 1.10 \\ 
nlp2&	10&	99.99&	GOpt&	8.73&	0.02&	32&	99.99&	GOpt&	0.34&	0.02 \\
nlp3&	10&	52.86&	0.004&	9.06&	TO&	10&	99.84&	GOpt&	59.00&	0.50 \\
ex1223a& 10& 99.00& GOpt& 6.31& 0.01& 10& 99.00& GOpt& 0.11& 0.01 \\
ex1264& 10&	39.72&	GOpt&	10.96&	1.48&	16&	40.56&	GOpt&	5.33&	1.74 \\
ex1265&	4&	23.74&	GOpt&	10.56&	0.64&	32&	23.74&	GOpt&	3.20&	0.72 \\
ex1266&	2&	82.29&	GOpt&	15.15&	0.02&	4&	82.29&	GOpt&	4.16&	0.34 \\
fuel& 4&	99.90&	GOpt&	6.95&	0.08&	4&	99.90&	GOpt&	0.14&	0.08 \\
meanvarx& 10&	67.28&	0.004&	6.50&	12.93&	4&	84.09&	0.0066&	8.26&	395.23 \\
util& 10&	99.99&	GOpt&	13.29&	0.47&	10&	99.99&	GOpt&	4.73&	0.83 \\
eniplac& 4&	19.15&	GOpt&	16.71&	49.83&	32&	19.15&	GOpt& 11.50&	5.56 \\

blend029& 32&	16.08&	GOpt&	15.73&	1.63&	10&	36.34&	GOpt&	4.76&	1.48 \\

blend531& 32&	6.91&	GOpt&	93.91&	63.77&	4&	9.48&	GOpt&	310.36&	82.09 \\

blend718& 16&	2.38&	GOpt&	52.07&	435.90&	32&	2.94&	GOpt&	28.46&	179.40\\

blend480&	16&	13.89&	0.092&	183.45&	1962.90&	16&	18.47&	0.097&	1014.47&	1029.00 \\

blend146& 32&	0.16&	0.043&	63.71&	TO&	8&	0.45&	0.057&	30.64&	TO \\
\hline
\specialrule{1.1pt}{1pt}{1pt}
\end{tabular}}
\label{tab:results}
 \end{table}
%\vspace{-0.5cm}

\noindent \textbf{Performance of CP/TCP}
Commonly in optimization adding extra binary variables increases problem complexity. However, in Table \ref{tab:results}, we observed that the run times for TCP were faster on twelve out of sixteen (including NLPs) instances. Blend480 was an exception, where TCP was almost five times slower than CP. Blend480 is one of the harder MINLPs; it has a large number of binary variables and constraints. From a total domain reduction (BC\%) perspective, the advantages of TCP are evident in Table \ref{tab:results}. Nlp3, meanvarx and blend029 have the largest reduction. The small BC\% values on ``blend" problems are due to variable bounds that are tight to begin with.

\noindent \textbf{Performance of convex relaxations on monomials}
Table \ref{tab:quad_relax} describes the performance of the algorithms when McCormick relaxations ($\langle x,x\rangle^{DTMC}$) are applied to monomial terms.
These results are compared with the tighter convex relaxations ($\langle x\rangle^{DTMC-q}$) of Table \ref{tab:results}. The run times of DTMC with tighter convex relaxations are faster on all the instances (best on eniplac). Moreover, the total reduction in bounds on variables during CP/TCP steps are up to 11\% larger using tighter convex relaxations. 
\vspace{-0.5cm}

\begin{table}
\footnotesize
\centering
\caption{Performance of algorithms with basic McCormick relaxations on higher-order monomials.}
\begin{tabular}{lrrrrrrrrrr} \toprule
Instances with  &\multicolumn{2}{c}{DTMC}
  &\multicolumn{4}{c}{CP-DTMC} & \multicolumn{4}{c}{TCP-DTMC}\\
  \cmidrule(lr){2-3}
  \cmidrule(lr){4-7}
  \cmidrule(lr){8-11}
  monomials & \%Gap &T$_{DTMC}$ & BC(\%)  & \%Gap  &T$_{CP}$ & T$_{DTMC}$ & BC(\%)  & \%Gap  &T$_{TCP}$ & T$_{DTMC}$\\
\midrule
  nlp1&0.0002&0.86&96.67&0.0002&8.16&0.98&98.89& 0.00013&1.64&0.37\\
  nlp2&GOpt&13.50&99.99&GOpt&7.99&2.49&99.99&GOpt&0.31&3.74\\
  ex1223a&0.0002&2.16&99.00&0.0001&5.84&0.31&99.00&0.0001&0.79&0.12\\
  fuel&GOpt&1.48&99.82&GOpt&6.45&0.20&99.90&GOpt&3.49&0.21\\
  meanvarx&0.012&755.86&67.28&0.0097&6.43&382.66&74.08&0.0077&6.63&453.31\\
  eniplac&0.0012&350.94&7.68&GOpt&20.31&2662.94&17.26&GOpt&29.98&68.21\\
\bottomrule
\end{tabular}
\label{tab:quad_relax}
\end{table}

\vspace{-1.2cm}
% Conclusions

\section{Conclusions}
\label{sec:conclusions}

\vspace{-.2cm}

In this work, we developed an approach for dynamically partitioning McCormick relaxations of multi-linear terms in MINLPs. This is a class of well-known, hard, non-convex optimization problems, where the lower bounds from these relaxations can be arbitrarily bad. We show that a dynamic partitioning of the domains of variables outperforms uniform partitioning and leads to a significantly smaller number of binary variables. We also show that CP techniques, such as bound contraction, can be applied in conjunction with dynamic partitioning to improve %computational performance dramatically. 
convergence drastically. 
Our numerical experiments suggest that the initial bounds of many benchmark problems are unnecessarily loose, lead to solver scalability issues, and result in poor relaxations. 

Finally, we emphasize that the algorithm presented in this paper is by no means exhaustive and there are a number of interesting directions for future research. First, the concept of dynamic partitioning could be combined with tighter convex over and under estimators for nonlinear functions and further improve the quality of the relaxations. Second, we only applied the bound tightening procedure at the root node. We could further apply it at sub nodes not unlike how \cite{mouret2009tightening} applies McCormick tightening. Third, there are CP propagation techniques that could be applied to further tighten variable domains. Finally, we could also improve the overall quality of the McCormick relaxations by using different orderings of variables in multi-linear terms. \\

%\vspace{-.1cm}

%\section*{Acknowledgements}
%\vspace{-0.5cm}
\noindent \textbf{Acknowledgements} The work was funded by the Center for Nonlinear Studies (CNLS) and was carried out under the auspices of the NNSA of the U.S. DOE at LANL under Contract No. DE-AC52-06NA25396. 

%\vspace{-.1cm}

%==============================;
%  Include all the references  ;
%==============================;
\bibliographystyle{splncs03}
\bibliography{references.bib}

\end{document}